\documentclass[11pt]{article}
\usepackage{fullpage}
\usepackage{amsmath, amsthm, amssymb}
\usepackage{graphicx}
\usepackage[latin1]{inputenc}
\usepackage{color}
\usepackage{bm}
\usepackage{algorithmicx}
\usepackage[noend]{algpseudocode}
\usepackage{algorithm}

\usepackage{pdfpages}

\newtheorem{corollary}{Corollary}

\newtheorem{lemma}{Lemma}

\newtheorem{theorem}{Theorem}

\newtheorem{prop}{Proposition}

\newtheorem{fact}{Fact}
\newtheorem{definition}{Definition}

\newcommand{\R}{\mathbb{R}}

\newcommand{\remove}[1]{}

\newcommand{\vvec}{\mathbf{v}}

\title{Information theoretical clustering is hard to approximate
\thanks{This paper was presented in part at the 36th International 
Conference on Machine Learning ICML 2019 \cite{CicLabicml19}.}}
\author{Ferdinando Cicalese\\  University of Verona, Italy\\
{\tt ferdinano.cicalese@univr.it} \and 
Eduardo Laber \\PUC-Rio, Brazil\\ {\tt laber@inf.puc-rio.br}}

\date{}

\begin{document}
\maketitle

\begin{abstract}
An impurity measures  $I: \mathbb{R}^d \mapsto  \mathbb{R}^+$ is a function that assigns a 
$d$-dimensional vector $\vvec$ to a non-negative value $I(\vvec)$ so that the more homogeneous $\vvec$, 
with respect to the values of its coordinates, the larger its impurity. A well known example of impurity measures is the 
entropy impurity. We study the problem of clustering based on the entropy impurity measures. 
Let $V$ be a collection of $n$ many $d$-dimensional vectors  with non-negative components. Given $V$ and an impurity measure
 $I$, the goal is to find a  partition ${\cal P}$ of $V$ into $k$  groups  $V_1,\ldots,V_k$   
 so as to minimize the  sum of the  impurities of the groups in 
 ${\cal P}$, i.e.,  $I({\cal P})= \sum_{i=1}^{k} I\bigg(\sum_{ \vvec \in V_i} \vvec \bigg).$

Impurity minimization has been widely  used as quality assessment measure in probability distribution clustering 
(KL-divergence) as well as in categorical clustering. However, in contrast to the case of 
metric based clustering, the current  knowledge of impurity measure based clustering in terms of approximation and inapproximability results is very limited.

Here, we contribute to change this scenario by proving that the problem of 
finding a clustering that minimizes the Entropy impurity measure 
is APX-hard, i.e., there
exists a constant $\epsilon > 0$ such that no polynomial time algorithm can 
guarantee $(1+\epsilon)$-approximation under the standard complexity hypothesis 
$P \neq NP$. The inapproximability holds even 
when all vectors have the same $\ell_1$ norm. 

This result  provides theoretical limitations on the computational  efficiency
that can be achievable  in the 
quantization of discrete memoryless channels, 
a problem that has recently attracted significant attention in the signal processing community.
In addition, it also solve a question that  remained open in previous 
work on this topic [Chaudhuri and McGregor COLT 08; Ackermann et. al. ECCC 11].

\end{abstract}

\section{Introduction}
Data clustering is a fundamental tool in machine learning
that is commonly used  to coherently organize data as well as 
to reduce the computational resources required to 
analyse  large datasets.
For  comprehensive descriptions of different clustering methods and their 
applications  refer to  \cite{HMMR15,Jain:1999}.
In general, clustering is the problem of partitioning a set of items so that, in the output partition, 
similar items are grouped together and 
dissimilar items are separated. 
When the items are represented as vectors that correspond
to frequency counts or  probability distributions
many clustering algorithms rely on 
so called impurity measures (e.g., entropy) that estimate the
dissimilarity of a group of items (see, e.g., 
 \cite{dmk-difcatc-03} and references therein).
In a simple example of this setting, a company may want to group users according to their taste for different genres of movies. Each user $u$ is represented by a vector, where the value of the $i$th component counts the number of times $u$ watched movies from  genre $i$.
To evaluate the dissimilarity of a group of users we calculate the impurity of the sum
of their associated  vectors, and then we select the partition for which the sum
of the dissimilarities of its groups is minimum.
The design of clustering methods based on  impurity measures is the central theme of this paper.

\subsection{Problem Description.}
An impurity measure  $I: \vvec \in \mathbb{R}_{+}^d \mapsto I(\vvec) \in \mathbb{R}_{+}$ is 
a function that maps a vector $\vvec$ to a non-negative value $I(\vvec)$ so that 
the more homogeneous $\vvec$ with respect to the values of its coordinates, the larger its impurity.
One of the most popular/studied impurity measures,  
henceforth referred to as the {\em Entropy impurity}
is based on the Shannon entropy function and it is defined by

$$ I_{Ent}(\vvec) = \|\vvec \|_1  \sum_{i=1}^d \frac{v_i}{\| \vvec\|_1} \log \frac{\| \vvec\|_1}{ v_i} $$

In the {\sc Partition with Minimum
Weighted Impurity Problem} ({\sc PMWIP}),
we are given a collection of $n$ many $d$-dimensional vectors $V \subset \mathbb{R}^d$ with non-negative
 components and we are also given an impurity measure $I$. The goal is to
find a  partition ${\cal P}$ of $V$ into $k$ disjoint groups of 	 vectors $V_1,\ldots,V_k$   so as to minimize
the  sum of the  impurities of the groups in ${\cal P}$, i.e.,  
\begin{equation}
\label{eq:goal}
I({\cal P})= \sum_{i=1}^{k} I\bigg(\sum_{ \vvec \in V_i} \vvec \bigg).
\end{equation} 

In this paper, our
focus is on the Entropy impurity $I_{Ent}$ as defined above.
We use {\sc PMWIP}$_{Ent}$  to refer to 
{\sc PMWIP} with impurity measure $I_{Ent}$.

\subsection{Applications}
This kind of clustering is used in a number of relevant
applications such as: (i)  attribute selection during the construction of random forest/decision trees \cite{Breiman84,BPKN:92,Chou91,journals/datamine/CoppersmithHH99,conf/icml/LMP18}; (ii) clustering of words based on their
distribution over a text collection  for improving classification tasks \cite{baker98distributional,dmk-difcatc-03} and (iii) quantization of memoryless channels/design of  polar codes \cite{journals/tit/TalV13,journals/tit/KurkoskiY14,journals/corr/KartowskyT17a,journals/tit/PeregT17,conf/NazerOP17}.
The last application  will be discussed more
thoroughly in the related work section.

Another interesting application, also mentioned in \cite{cm-fmsitc-08,journals/talg/AckermannBS10},  
is the compression of  a large collection of $n$ short files (e.g. tweets) using
entropy encoding (e.g. Huffman or Arithmetic coding).
Compressing  each file individually might incur
a huge overhead since we would need  to store
the compression model (e.g. alphabet + codewords) for each of them. Thus, a natural idea is to cluster them into $k \ll n$
groups and then compress files in the  same cluster using the same model.
This 
 approach leads to an instance of   {\sc PMWIP}$_{Ent}$  
since it is possible to  generate encodings for
 probability  distributions whose sizes
 are arbitrarily close to their Shannon entropy, 
 so that the objective function in (\ref{eq:goal}) is the total size of the compressed collection  ignoring the compression model for each of the clusters.

Despite of its wide use  in relevant applications, clustering being a fundamental problem
and entropy being an essential  measure in  Information Theory as well as considerably important in Machine Learning, the current  understanding of  {\sc PMWIP}
from the  perspective of approximation algorithms 
is very limited as we detail further.
 This contrasts with  what is known for  clustering in metric spaces
 where  the gap between the ratios
achieved by the best known  algorithms and the largest known inapproximability 
factors, assuming $P\ne NP$, are somehow tight (see \cite{journals/corr/AwasthiCKS15} and references therein).  Our study contributes 
to reducing this gap of knowledge.

\subsection{Our Results and Techniques} 
Our main contribution is a proof  that
{\sc PMWIP}$_{Ent}$  is APX-Hard
even for the case where all vectors have the same $\ell_1$-norm. In other words, we show that there exists $\epsilon >0$ such that one cannot obtain in polynomial time an $(1+\epsilon)$  multiplicative  approximation for {\sc PMWIP}$_{Ent}$ unless $P=NP$.
Our result implies on the
APX-Hardness of  $MTC_{KL}$, the problem of clustering $n$ probability distributions into $k$ groups where the Kullback-Leibler divergence is used to measure the distance between each distribution and the centroid of its assigned group. 
With this result, we settle a question that remained open in previous work on the $MTC_{KL}$ problem \cite{cm-fmsitc-08,journals/eccc/AckermannBS11}.
Finally, our inapproximability result  contributes to the understanding of a problem that has recently attracted significant attention in the signal processing community, as it provides a theoretical limit, under the perspective of computational
complexity, on the possibility of efficiently quantizing discrete
memoryless channels.  In fact, as we explain in Section \ref{sec:quantizers}, the quantization problem can
be formulated as  the {\sc PMWIP}$_{Ent}$.

Our proof follows the approach employed by Awasthi  
\cite{journals/corr/AwasthiCKS15}
to show  the APX-hardness of the $k$-means clustering problem.
As in \cite{journals/corr/AwasthiCKS15}, it consists of mapping hard instances from the minimum vertex cover problem in triangle-free bounded degree graphs into
instances of the clustering problem, in our case the
{\sc PMWIP}$_{Ent}$.
However, in order to deal with  the entropy impurity, we have to
overcome a number of technical challenges that do not arise when one considers the $\ell_2^2$ distance employed by the $k$-means problem.

\subsection{Paper Organization}
The paper is organized as follows:
in Section \ref{sec:RelatedWork} we discuss 
works
that are  related to ours. More specifically,
we discuss existing work on the problem of quantizing
discrete memoryless channels and on the the problem
of clustering probability distributions using Kullback-Leibler divergence.
In Section \ref{sec:Prelim} we present some technical background that is
useful to derive our main result.
Section \ref{inapprox} is dedicated to our main contribution. 
Section \ref{sec:conclusions} concludes the paper with some 
final remarks and open questions.

\section{Related Work}
\label{sec:RelatedWork}

\subsection{Connections to quantizer design}
\label{sec:quantizers}
Quantization is a fundamental part of digital processing systems used to keep complexity and resource consumption tractable.
Quantization refers to the mapping of a large alphabet to a smaller one. A typical example is given by 
making the output of a channel $X \mapsto Y$  with a large output alphabet travel through an additional channel $Y \mapsto Z$ 
(the quantizer) whose output is from a small alphabet without affecting too much
 the capacity of the 
quantized channel $X \mapsto Z$ with respect to the capacity of the original channel $X \mapsto Y$.

To formalize the connection between our problem and the one of designing optimal quantizers, let us consider a
discrete memory channel (DMC). Let the input be $X \in {\cal X} = \{1, \dots, d\}$ with distribution $p_x = Pr(X = x).$ 
Let the channel output be $Y \in {\cal Y} = \{1, \dots, n\}$ and the channel transition probability be denoted by 
$P_{y\mid x} = Pr(Y = y \mid X = x).$
The channel output is quantized to $Z \in {\cal Z} = \{1, \dots, k\}$ for some $k \leq n,$ by 
a possibly stochastic channel (quantizer) defined by $Q_{z\mid y} = Pr(Z = z \mid Y = y),$ 
so that the conditional probability distribution of the quantizer output $Z$ w.r.t. the channel input $X$ is 
$T_{z \mid x} = Pr(Z = z \mid X = x) = \sum_{y \in Y} Q_{z\mid y} P_{y \mid x}.$

The aim of quantizer's design is to guarantee that the mutal information $I(X; Z)$ between $X$ and $Z$ is as large as possible,
i.e., as close as possible 
to the mutual information $I(X; Y)$  between the input and the non quantized output $Y$.

It is known that there is an  optimal quantizer that is deterministic (see, e.g., \cite{journals/tit/KurkoskiY14}), i.e., for each 
$y \in Y$ there is a unique $z \in Z$ such that $Q_{z \mid y} = 1.$
Therefore, the quantizer design problem can be cast as finding the $k$-partition ${\cal C} = C_1, \dots, C_k$ of ${\cal Y}$ such that 
$T_{z \mid x} = \sum_{y \in C_z} P_{y \mid x}$ guarantees the maximum mutual information 
$$I(X; Z) = \sum_{z \in {\cal Z}}  \sum_{x \in {\cal X}}  p_x T_{z \mid x} \log  \frac{T_{z \mid x}}{\sum_{x'} p_{x'} T_{z \mid x'}}.$$

\medskip

If for each $y \in {\cal Y}$ we define the vector $\vvec^{(y)} = (v^{(y)}_1, \dots, v^{(y)}_d)$ with $v^{(y)}_x = p_x P_{y \mid x}$ and  
let ${\bf c}^{(z)} = \sum_{y \in C_z} \vvec^{(y)},$ we have that $$T_{z \mid x} = \sum_{y \in C_z} P_{y \mid x} = \sum_{y \in C_z} \frac{v^{(y)}_x}{p_x}.$$
Hence, 
\begin{eqnarray*}
I(X, Z) &=& \sum_{z} \sum_{x} p_x \sum_{y \in C_z} \frac{v^{(y)}_x}{p_x} 
\log \left( \frac{\sum_{y \in C_z} \frac{v^{(y)}_x}{p_x} }{\sum_{x'} p_{x'} \sum_{y \in C_z} \frac{v^{(y)}_{x'}}{p_{x'}}} \right) \\
&=& \sum_{x} \sum_{z} c^{(z)}_x \log \frac{c^{(z)}_x}{p_x \| {\bf c}^{(z)}\|_1} \\
&=& \sum_x \sum_{z} c^{(z)}_x \log \frac{1}{p_x}  - \sum_z \|{\bf c}^{(z)} \|_1 \sum_x \frac{c^{(z)}_x}{\| {\bf c}^{(z)} \|_1} \log \frac{\| {\bf c}^{(z)} \|_1}{c^{(z)}_x} \\
&=& \sum_x p_x \log \frac{1}{p_x} - \sum_z I_{Ent}(C_z) = H(X) - \sum_{z=1}^k I(C_z).
\end{eqnarray*}

Therefore, an optimal quantizer that maximizes the mutual information, is obtained by the clustering of minimum impurity
$$\sup_{Q_{z \mid y}} I(X; Z) = H(X) - \inf_{(C_1, \dots C_k)} \sum_{z = 1}^k I(C_z).$$

Efficient techniques for maximizing mutual information in channel output quantization have been the focus 
of several recent papers as the problem plays an important role in making the implementation of 
polar codes tractable \cite{journals/tit/Tal17,DBLP:conf/isit/PedarsaniHTT11,conf/isit/TalSV12,journals/corr/KartowskyT17a}. 
The objective in these papers is the minimization of the additive loss 
in terms of mutual information due to the quantization, i.e., the minimization/approximation of the objective function $I(X,Y) - I(X,Z)$ that 
(by proceeding as above) is equivalent, in our notation,  to the minimization/approximation of 
$$\Delta_k(X;Y) = I(X,Y) - I(X, Z) = \sum_{z} I(C_z) - \sum_{y} I(v^{(y)}).$$

A main result of \cite{journals/corr/KartowskyT17a} is that for arbitrary joint distribution $P_{XY}$ it holds that 
$\Delta_k(X;Y) = O(k^{-2/(d-1)}).$ In \cite{journals/tit/Tal17} it is shown that 
there exist channels $X\mapsto Y$ such that $\Delta_k(X;Y) = \Omega(k^{-2/(d-1)}).$
The papers \cite{DBLP:conf/isit/PedarsaniHTT11,conf/isit/TalSV12,journals/corr/KartowskyT17a} also 
provided polynomial-time approximation algorithm for designing such quantizers. In addition, 
for binary input channels ($|{\cal X}| = 2$)  an optimal algorithm is given in \cite{journals/tit/KurkoskiY14} and 
\cite{DBLP:conf/isit/IwataO14}.

\medskip
We note that, although the optimal clustering provides an optimal quantizer, in terms of approximation, because of the 
additional term in the objective function $\Delta_k(X;Y)$ we have that an approximation guarantee for $\Delta_k(X;Y)$
also gives an approximation guarantee on the optimal impurity, but not vice versa. 
On the other hand, our inapproximability result for {\sc PMWIP}$_{Ent}$ 
implies the same inapproximability result for the problem of optimizing $\Delta_k(X;Y).$

\medskip

A different perspective has been taken in \cite{conf/NazerOP17} where the multiplicative loss $I(X;Z)/I(X;Y)$ has been studied, with the 
aim of characterizing the fundamental properties of the  quantizer attaining the minimum $I(X;Z)$. In particular the authors of \cite{conf/NazerOP17} 
study the 
infimum of $I(X;Z)$ taken with respect to all joint distributions with input alphabet ${\cal X}$ of cardinality $d$ and {\em arbitrary} (possibly continuous)
output alphabet ${\cal Y}$ such that the mutual information $I(X;y)$ is at least $\beta,$ where $\beta$ is a given parameter.

\subsection{Clustering with Kullback-Leibler Divergence} 
{\sc PMWIP}$_{Ent}$ is closely related to  $MTC_{KL}$ defined in 
 \cite{cm-fmsitc-08} as  the  problem of  
 clustering a set of $n$ probability distributions 
into $k$ groups minimizing the total Kullback-Leibler (KL) divergence 
from  the distributions to the centroids of their assigned groups.
Mathematically, we are given a set of $n$ points $p^{(1)},\ldots,p^{(n)}$,
corresponding to probability distributions, and a positive 
integer $k$. The goal is to find a partition of the points
into $k$ groups $V_1,\ldots,V_k$  and a centroid $c^{(i)}$ for each group
$V_i$ such that
$$\sum_{i=1}^k \sum_{p \in V_i} KL(p, c^{(i)}) $$ is
minimized, where $KL(p,q)=\sum_{j=1}^d p_j \ln (p_j/q_j)$ is the Kullback-Leibler
divergence between points $p$ and $q$.

It is known that in the  optimal solution  for each $i=1, \ldots k$ the centroid $c^{(i)} = (c^{(i)}_1,\dots c^{(i)}_d)$ is given by
 $c^{(i)}_j=\sum_{p \in V_i}  p_j / |V_i|$, for each $j=1,\ldots,d$. Thus, 
$MTC_{KL}$  is equivalent to the problem of finding 
a  partition  that minimizes
\begin{align}
\sum_{i=1}^k \sum_{p \in V_i} KL(p, c^{(i)}) = \sum_{i=1}^k \sum_{p \in V_i}  \sum_{j=1}^d p_j (\ln p_j -\ln c_j^{(i)}) = \\
  \sum_{i=1}^n \sum_{j=1}^d p^{(i)}_j \ln  p^{(i)}_j - \sum_{i=1}^k \sum_{p \in V_i} \sum_{j=1}^d  p_j \ln c^{(i)}_j  = \\ 
 \sum_{i=1}^n \sum_{j=1}^d p^{(i)}_j \ln  p^{(i)}_j 
- \sum_{i=1}^k  \sum_{j=1}^d  \left ( \sum_{p \in V_i}   p_j \right)  \ln
\frac{ \left ( \sum_{p \in V_i}   p_j \right) }{|V_i|} = \\
- \frac{1}{\log e} 
\sum_{i=1}^{n} I_{Ent}( p^{(i)}) +
\frac{1}{\log e} \sum_{i=1}^{k} I_{Ent}\bigg(\sum_{ p \in V_i} p \bigg) 
\end{align}

Therefore, the optimal solution of $MTC_{KL}$ is equal 
to  the optimal one of the particular case of {\sc PMWIP}$_{Ent}$ in which  $\vvec_i=p_i$ for $i=1,\ldots,n$.
While their optimal solutions match in this case,  {\sc PMWIP}$_{Ent}$ and $MTC_{KL}$  differ in terms of approximation  since 
the  objective function for $MTC_{KL}$ has an additional constant
term  $- \sum_{i=1}^n   I_{Ent} (p^{(i)})$ so that an $\alpha$-approximation for $MCT_{KL}$ problem implies an $\alpha$-approximation for  {\sc PMWIP}$_{Ent}$ while the converse is not necessarily true.

In \cite{cm-fmsitc-08} an $O(\log n)$ approximation for {\sc MTC}$_{KL}$
 is given. 
Some $(1+\epsilon)$-approximation algorithms
were proposed for a constrained version of {\sc MTC}$_{KL}$
where  every element of every probability distribution lies
in the interval $[\lambda,v]$  \cite{Ackermann:2008:CMN:1347082.1347170,conf/soda/AckermannB09,journals/talg/AckermannBS10,Lucic16}.
The algorithm from   \cite{Ackermann:2008:CMN:1347082.1347170,journals/talg/AckermannBS10} runs
in $O(n 2^{O( m k / \epsilon  \log (m k / \epsilon) ) }  )$ time,
where $m$ is a constant that depends
on $\epsilon$ and $\lambda$.  In  \cite{conf/soda/AckermannB09}
the running time is improved to
 $O(ngk + g 2^{O(k / \epsilon 
\log (k / \epsilon) )} \log^{k+2} (n))$ via the use of weak coresets.
Recently, using strong coresets, $O(ngk + 2^{poly(gk / \epsilon)}$ time is obtained \cite{Lucic16}.
We shall note that these algorithms provide guarantees for $\mu$-similar Bregman divergences, a class of metrics that includes domain constrained $KL$ divergence.  
By using similar assumptions on the components of the input probability distributions, 
Jegelka et. al.  \cite{journals/corr/abs-0812-0389} show that Lloyds 
$k$-means algorithm---which also has an  exponential time  worst case complexity
\cite{Vattani2011}---obtains an $O( \log k)$ approximation for
{\sc MTC}$_{KL}$. For {\sc PMWIP}$_{Ent}$, 
 an $O(\log^2 \min\{k,d\})$ is presented in \cite{CicLabicml19}.

In terms of computational complexity,  
Chaudhuri and McGregor \cite{cm-fmsitc-08} proved
that the variant of $MTC_{KL}$  where
the centroids must be chosen from the input probability distributions is
NP-Complete. The NP-Hardness of  $MTC_{KL}$, that remained open in  \cite{cm-fmsitc-08}, 
 was established in  Ackermann et. al. \cite{journals/eccc/AckermannBS11},
where it is also mentioned
that the APX-hardness of $k$-means in $\R^2$ would imply the same kind of hardness for
$MTC_{KL}$. However, it is not known whether the former is APX-Hard.
Our result provides an important progress in this line of investigation
since it establishes the APX-hardness of $MTC_{KL}$.
We shall mention that, in terms of restricted instances, 
 one of the authors proved recently  that {\sc PMWIP}$_{Ent}$ is $NP$-Complete, even when $k=2$, via a simple  reduction from {\sc Partition} \cite{conf/icml/LMP18}.

\section{Preliminaries}
\label{sec:Prelim}
In this section we introduce some notation and
discuss some technical material that 
will be useful to establish our main result.
We start with some notation.

\subsection{Notation}

Let $G=(V,E)$ be a simple undirected graph.
A vertex cover in $G$  is a set
of vertices $S$ such that each edge $e \in E$ is incident to some vertex of $S$. 
A vertex cover $S$ is minimal if for every $v \in S$,  $S \setminus \{v\}$ is not a 
vertex cover. 
A star in $G$ is a subgraph $G'$ of $G$ such that there
exists one vertex in $G'$, the centre of the star $G'$,
that is incident to all edges in $G'$. If a star $G'$ has
$p$ edges we say that it is a $p$-star.

We say that $G$ is a $D$-bounded degree graph if for each vertex $v \in V$ the degree of $v$ is at most $D$. We say that 
$G$ is $D$-regular, if for each $v \in V$ the degree of $v$ is $D$. 
A triangle in $G$ is a set of three vertex $u,v$ and $w$ such that
$uv,uw,vw \in E$.
We say that $G$ is triangle-free if it has no triangle.

For a set of vectors $C$ we use $I_{Ent}(C)$ to denote the impurity  of $C$, that is, 
$I_{Ent}(C)=I_{Ent}( \sum_{\vvec \in C} \vvec)$.

We define
the entropy impurity of a set of edges $C$ from a graph $G$, over the vertex set $V = \{v_1, 
\dots,  v_n\}$, as follows:
\begin{equation} \label{pstarentro}
I_{Ent}(C) = 2|C| \cdot H \left  (\frac{d^C(v_1)}{2|C|}, \dots, \frac{d^C(v_n)}{2|C|} \right)
\end{equation}
where $H()$ denotes the Shannon entropy and 
$d^C(v)$ denotes the number of edges of $C$ incident in $v$

\subsection{Bounds on the Entropy Impurity for Graphs} 

We now present bounds on
the entropy impurity of sets of edges focussing in particular on subsets of edges from 
triangle-free graphs, which will be the basis of our hardness proof in the next section.

\begin{fact} \label{fact:star-impurity}
If $C$ is a set of edges forming a $p$-star, then we have 
$I_{Ent}(C) = 2p + p \log p$.
\end{fact}

We also have that the $3$-bounded triangle-free graphs 
with $p$-edges have impurity at least 
$2p + p \log p$  as  recorded in  the following Lemma \ref{fact:starLB}.

To establish the lemma we will exploit  Facts \ref{fact:Schur} 
and  \ref{fact:Mantel} presented in the sequel. The former 
is a direct consequence of the Schur concavity of the Shannon entropy function (see, e.g., \cite{marshall11}) while the latter is a special case of Turan's and Mantel's theorem about triangle-free graphs (see, e.g., \cite{Jukna10}).

\begin{fact} \label{fact:Schur}
Let ${\bf p} = (p_1, \dots p_n)$ and ${\bf q} = (q_1, \dots q_n)$ be probability distributions such that 
for each $i=1, \dots, n-1$ it holds that 
\begin{itemize}
\item $p_i \geq p_{i+1}$;
\item $q_i \geq q_{i+1}$;
\item $\sum_{j=1}^i p_i  \leq \sum_{j=1}^i q_i.$
\end{itemize}
Then, $H({\bf p} ) \geq H({\bf q} ).$
\end{fact} 

\begin{fact} \label{fact:Mantel}
Let $n$ and $m$ denote the number of vertices and edges of a triangle-free graph. Then $n \geq \lceil \sqrt{4m} \rceil.$
\end{fact}

\begin{lemma} \label{fact:starLB} 
Let $C$ be a set of $p$ edges from a $3$-bounded degree triangle-free graph. Then $I_{Ent}(C) \geq 2p + p \log p.$
\end{lemma}
\begin{proof}
 Recalling the definition in (\ref{pstarentro}) with $|C| = 2p$ and denoting by 
$d^C(v)$ the number of edges of $C$ incident in $v$ and by $v_1, \dots, v_n$ the 
vertices of the underlying graph, we have 
\begin{equation} \label{pstarentro200}
I_{Ent}(C) = 2p H \left  (\frac{d^C(v_1)}{2p}, \dots, \frac{d^C(v_n)}{2p} \right),
\end{equation}
where $H()$ denotes the Shannon entropy. Since the entropy function is invariant upon permutations of the components, for the rest of this proof, 
we will assume w.l.o.g., that the vertices are sorted in non increasing order of degree, i.e., $d^C(v_i) \geq d^C(v_{i+1})$ for $i=1, \dots, n-1.$ 
Let $\tilde{n}$ denote the number of vertices incident on edges in $C$.
Equivalently, under the standing assumption, $\tilde{n}$ is the largest index $i$ such that $d^C(v_i) > 0.$

The desired result will follow from showing that the entropy on the right hand side of (\ref{pstarentro200}) is lower bounded by $1 + \frac{\log p}{2}.$
We will split the analysis into several cases. 

\medskip
\noindent
{\em Case 1.} $ p \geq 9.$

Since the edges in $C$ are from a $3$-bounded degree graph, we have  
$d^C(v_i) \leq 3$ for $i=1, \dots, n.$ Hence, 
$$ H \left  (\frac{d^C(v_1)}{2p}, \dots, \frac{d^C(v_n)}{2p}\right) = \sum_{j=1}^{\tilde{n}} \frac{d^C(v_i)}{2p} \log \frac{2p}{d^C(v_i)} \geq 
 \sum_{j=1}^{\tilde{n}} \frac{d^C(v_i)}{2p} \log \frac{2p}{3}  = 1 + \log \frac{p}{3} \geq 1+ \frac{1}{2} \log p,$$
 where, in the first inequality we are using $d^{C}(v_i) \geq 3$ and the last inequality follows from $p \geq 9.$

\medskip
\noindent
{\em Case 2.} $\tilde{n} \geq p+1.$

By the standing assumptions, we have $d^C(v_1) \geq \cdots \geq d^C(v_{p+1}) > 0.$ Therefore, by Fact \ref{fact:Schur}, 
we have that 
 \begin{equation} \label{pstarentro2}
H\left(\frac{d^C(v_1)}{2p}, \dots, \frac{d^C(v_n)}{2p}\right) \geq 
H\left(\frac{p}{2p}, \frac{1}{2p}, \dots \frac{1}{2p}, 0, \dots, 0 \right) = 1 + \frac{1}{2} \log p.
\end{equation}

\medskip
By Fact \ref{fact:Mantel}, for $p = 3$, we have that the triangle-free condition implies $\tilde{n} \geq 4 = p+1$, hence 
the last case also covers  $p=3.$

\medskip

Under the condition in Fact \ref{fact:Mantel}, it remains to consider the cases $4 \leq p \leq 8$ with $\lceil \sqrt{4p} \rceil \leq \tilde{n} \leq p,$ i.e., 
(i) $p = 4, \tilde{n} = 4,$ ~(ii) $p = 5, \tilde{n} = 5,$ ~(iii) $p = 6, \tilde{n} \in \{5, 6\},$ ~(iv) $p = 7, \tilde{n} \in \{6,7\},$ and 
 ~(v) $p = 8, \tilde{n} \in \{6,7,8\}.$
For each one of these cases we shall show a 
probability distribution ${\bf w}^{(p,\tilde{n})}$  such that the inequality 
$H({\bf w}^{(p,\tilde{n})}) \geq 1 + \frac{\log p}{2}$ holds and 
by Fact \ref{fact:Schur} also $H({\bf w}^{(p,\tilde{n})}) \leq H\left(\frac{d^C(v_1)}{2p}, \dots, \frac{d^C(v_n)}{2p}\right)$ holds  for every 
choice of a triangle-free set $C$ with $p$ edges incident to $\tilde{n}$ vertices.

\bigskip
\noindent
{\em Case 3.} $p = 4, \, \tilde{n} = 4$. 
The only triangle-free graph on $\tilde{n} = 4$ vertices with $p = 4$ edges is 
given by a polygon with 4 vertices, i.e., $d^C(v_i) = 2,$ for $i=1, 2, 3, 4.$ In this case we have
$$H\left(\frac{d^C(v_1)}{2p}, \dots, \frac{d^C(v_4)}{2p}\right) = H\left(\frac{2}{8}, \frac{2}{8}, \frac{2}{8}, \frac{2}{8}\right) = 2 = 1 + \frac{\log p}{2},$$
as desired. 
 
\medskip
\noindent
{\em Case 4.} $p = 5, \, \tilde{n} \geq 5$.  
 
By by the assumptions $3 \geq d^C(v_1) \geq \cdots \geq d^C(v_{5}) > 0,$ and Fact \ref{fact:Schur}, we have that 
 \begin{equation} \label{pstarentro3}
H\left(\frac{d^C(v_1)}{2p}, \dots, \frac{d^C(v_n)}{2p}\right) \geq 
H \left (\frac{3}{10}, \frac{3}{10}, \frac{2}{10}, \frac{1}{10}, \frac{1}{10}, 0, \dots, 0 \right) =  
\log 5 + \frac{8}{10} - \frac{6}{10} \log 3 \geq 1 + \frac{\log 5}{2} 
\end{equation}
The desired result now follows by noticing that the last quantity 
in (\ref{pstarentro3})
is equal to $1 + \frac{1}{2} \log p.$

\medskip
\noindent
{\em Case 5.} $p = 6, \, \tilde{n} \in \{5, 6\}$. 

\begin{itemize}
\item[] {\em Subcase 5.1} $p = 6, \, \tilde{n} = 5$

 By direct inspection, it is not hard to see that there is no triangle-free graph on $5$ vertices where three of them have degree $3$. Hence, we have that
 the following conditions must hold
 $$3 \geq d^C(v_1) \geq d^C(v_2) > 0 \quad \mbox{and} \quad 2 \geq d^C(v_3) \geq d^C(v_4) \geq d^C(v_5) > 0.$$
 Therefore, by Fact \ref{fact:Schur}, we have  
 \begin{equation} \label{pstarentro400}
H \left  (\frac{d^C(v_1)}{2p}, \dots, \frac{d^C(v_n)}{2p}\right) \geq 
H \left (\frac{3}{12}, \frac{3}{12}, \frac{2}{12}, \frac{2}{12}, \frac{2}{12}, 0, \dots, 0 \right) =  
1 + \frac{\log 6}{2} 
\end{equation}
The desired result now follows by noticing that the last quantity 
in (\ref{pstarentro400})
is equal to $1 + \frac{1}{2} \log p.$

\item[] {\em Subcase 5.2} $\tilde{n} = 6$

By Fact \ref{fact:Schur} and the assumption $3 \geq d^C(v_1) \geq \cdots \geq d^C(v_{6}) > 0,$ we have  
 \begin{equation} \label{pstarentro5}
H \left  (\frac{d^C(v_1)}{2p}, \dots, \frac{d^C(v_n)}{2p}\right) \geq 
H \left (\frac{3}{12}, \frac{3}{12}, \frac{3}{12}, \frac{1}{12}, \frac{1}{12}, \frac{1}{12}, 0, \dots, 0 \right) =  
2 +  \frac{3}{12} \log 3 \geq 1 + \frac{\log 6}{2} 
\end{equation}
The desired result now follows by noticing that the last quantity 
in (\ref{pstarentro5})
is equal to $1 + \frac{1}{2} \log p.$

\end{itemize}

\medskip
\noindent
{\em Case 6.} $p = 7, \, \tilde{n} \in \{6,7\}$. 

By Fact \ref{fact:Schur}, and the assumption $3 \geq d^C(v_1) \geq \cdots \geq d^C(v_{6}) > 0,$ we have that 
 \begin{equation} \label{pstarentro4}
H \left  (\frac{d^C(v_1)}{2p}, \dots, \frac{d^C(v_n)}{2p}\right) \geq 
H \left (\frac{3}{14}, \frac{3}{14}, \frac{3}{14}, \frac{3}{14}, \frac{1}{14}, \frac{1}{14}, 0, \dots, 0 \right) =  
1 + \log 7 - \frac{6}{7} \log 3 > 1 + \frac{\log 7}{2}
\end{equation}
The desired result now follows by noticing that the last quantity 
in (\ref{pstarentro4})
is equal to $1 + \frac{1}{2} \log p.$

\medskip
\noindent
{\em Case 7.} $p = 8, \, \tilde{n} \in \{6, 7, 8\}$. 

By Fact \ref{fact:Schur} and the assumption $3 \geq d^C(v_1) \geq \cdots \geq d^C(v_{6}) > 0,$ we have that 
 \begin{equation} \label{pstarentro4000}
H \left  (\frac{d^C(v_1)}{2p}, \dots, \frac{d^C(v_n)}{2p}\right) \geq 
H \left (\frac{3}{16}, \frac{3}{16}, \frac{3}{16}, \frac{3}{16}, \frac{3}{16}, \frac{1}{16}, 0, \dots, 0 \right) =  
4 -  \frac{15}{16} \log 3 > 1 + \frac{3}{2}
\end{equation}
The desired result now follows by noticing that the last quantity 
in (\ref{pstarentro4000})
is equal to $1 + \frac{1}{2} \log p.$

\end{proof}

\section{Hardness of Approximation of {\sc PMWIP}$_{Ent}$} \label{inapprox}
The goal of this section is to establish our main result.

\begin{theorem} \label{theo-main}
{\sc PMWIP}$_{Ent}$ is APX-Hard.
\end{theorem}

This section is split into two subsections. In Section \ref{sec:gappres}
we present a gap preserving  
reduction from vertex cover to our clustering problem. In Section \ref{sec:gappres-correc} 
 we 
establish its correctness.

\subsection{A Gap Preserving Reduction}
\label{sec:gappres}

 We start by recalling some basic definitions and facts that are useful for establishing limits on the 
approximability  of optimization problems (see, e.g., \cite[chapter 29]{Vazirani:2001}).

Given a minimization problem $\mathbb{A}$ and a parameter $\epsilon>0$ 
we define the {\sc $\epsilon$-Gap}-$\mathbb{A}$ problem as the problem of 
deciding for an instance $I$ of $\mathbb{A}$ and a parameter $k$ 
whether: (i) $I$ admits a solution of 
value $\leq k$; or (ii) every solution of $I$ have value $>(1+\epsilon) k.$ In such a 
gap decision problem it is tacitly assumed that the instances are either of type (i)
or of type (ii).

\begin{fact} \label{fact:gap-inapprox}
If for a minimization problem $\mathbb{A}$ there exists $\epsilon > 0$ such that the 
{\sc $\epsilon$-Gap}-$\mathbb{A}$ problem is $NP$-hard, then no polynomial time
$(1+\epsilon)$-approximation algorithm exists for $\mathbb{A}$ unless $P = NP.$
\end{fact}

We will use the following definition of a gap-preserving reduction.
\begin{definition} \label{gap-reduction}
Let $\mathbb{A}, \mathbb{B}$ be minimization problems. A gap-preserving reduction from $\mathbb{A}$
to $\mathbb{B}$ is a polynomial time algorithm that, given an instance $x$ of $\mathbb{A}$ and a value $k$, produces 
an instance $y$ of $\mathbb{B}$ and a value $\kappa$ such that  there exist constants $\epsilon, \eta > 0$ for which 
\begin{enumerate}
\item if $OPT(x) \leq k$ then $OPT(y) \leq \kappa$;
\item if $OPT(x) > (1+\epsilon) k$ then $OPT(y) > (1+\eta) \kappa$; 
\end{enumerate}
\end{definition}

\begin{fact} 
Fix minimization problems  $\mathbb{A}, \mathbb{B}$. If there exists 
$\epsilon$ such that  the {\sc $\epsilon$-Gap}-$\mathbb{A}$ problem is $NP$-hard
and there exists a  gap-preserving reduction from $\mathbb{A}$
to $\mathbb{B}$ then there exists $\eta$ such that 
the {\sc $\eta$-Gap}-$\mathbb{B}$ problem is $NP$-hard
\end{fact}

 We will now specialize the above definitions for a restricted variant of 
 the problem of finding a minimum vertex cover in a graph and for our 
 clustering problem 
 {\sc PMWIP}$_{Ent}.$

\begin{definition}
For every $\epsilon > 0$, the
{\sc $\epsilon$-Gap-MinVC-4}
 (gap) decision problem is defined as follows:
given a $4$-regular graphs $G = (V, E)$ and an integer $k$, decide
whether $G$  has a vertex cover of size $k$ or all vertex covers of $G$ 
have size $> k (1+\epsilon)$.
\end{definition}

\begin{definition}
For every $\eta > 0$, the
{\sc $\eta$-Gap-PMWIP}$_{Ent}$
(gap) decision problem is defined as follows:
given a set of vectors $U$, an integer $k$, and a value $\kappa$,  decide
whether there exists a $k$-clustering ${\cal C} = (C_1, \dots C_k)$ of the vectors in $U$ such that the
total impurity $I_{Ent}({\cal C}) = \sum_{\ell= 1}^k I_{Ent}(C_{\ell})$ is  at most $\kappa$ or 
for each $k$-clustering ${\cal C}$ of $U$ it holds that $I_{Ent}({\cal C}) > (1+\eta)\kappa.$
\end{definition}

The following result is a consequence of \cite[Theorems 17 and 19]{Chlebik-ova06}.

\begin{theorem} \cite{Chlebik-ova06} \label{theo:4-reg-graphgap}
There are constants $0 < \alpha < \alpha' < 1$ 
and 4-regular graphs $G = (V, E)$ such that it is $NP$-Complete to decide whether  $G$ 
has a vertex cover of size $\alpha |V|$ or all vertex covers of $G$ have size $> \alpha' |V|$.
Hence for $\epsilon = \frac{\alpha'}{\alpha} - 1$, the {\sc $\epsilon$-Gap-MinVC-4} is NP-Complete.
\end{theorem}

In order to show the APX-hardness of {\sc PMWIP}$_{Ent}$, we employ a gap-preserving reduction from minimum vertex cover ({\sc MVC}) in 4-regular graphs to  
{\sc PMWIP}$_{Ent}.$
This reduction is obtained by combining: (i) a gap-preserving reduction  from 
{\sc MVC} in $4$-regular graphs to {\sc MVC} in $4$-bounded degree triangle-free graphs \cite{LEE201740};
(ii) an L-reduction from {\sc MVC} in $4$-bounded degree graphs to 
{\sc MVC} in $3$-bounded degree graphs \cite{Kann}; 
(iii) a gap-preserving reduction from instances of 
{\sc MVC} in 3-bounded degree triangle-free graphs obtained via (i)-(ii) to instances of 
{\sc PMWIP}$_{Ent}.$

We will first recall the maps at the bases of the  reductions (i) and (ii). The proofs that they define gap-preserving reduction (Theorem \ref{theo:cubicgraphgap})  will be included in 
the appendix  for the sake of self-containment. 

\bigskip

\noindent \underline{\bf From $4$-regular graphs to $4$-bounded degree triangle-free graphs.} 
\label{sec:r1}
Let $G' = (V', E')$ be a $4$-regular graph. Let $n = |V'|$ and $m = |E'|$. Hence $m = 2n.$
Since in every graph there is a cut containing at least half of the edges, we can select 
$\hat{E} \subseteq E'$ with $|\hat{E}| = m/2 = n$ such that the graph with vertex set $V'$ and set of edges $\hat{E}$  is bipartite and, as a consequence,  triangle-free. 
For each $e = (u,v) \in E' \setminus \hat{E}$ define a new set of vertices $V_e = \{u', v'\}$ and let 
$$V_H = V' \cup \bigcup_{e = (u,v) \in E' \setminus \hat{E}} V_e \qquad E_H = \hat{E} \cup  \bigcup_{e = (u,v) \in E' \setminus \hat{E}} \{(u, u'), (u', v'), (v', v)\}$$.

Finally let $H = (V_H, E_H)$. In words, $H$ is obtained from $G'$ by a double subdivision of  the $n$ edges not in $\hat{E}$. 

By construction we have that the graph $H$ has bounded degree $4$ and it is also triangle-free.

\medskip

Let ${\cal R}^{(1)}$ be a function that maps a $4$-regular graph $G'$ to a $4$-bounded degree triangle-free graph $H$ according to the procedure defined above. 

\bigskip

\medskip
\noindent \underline{\bf From $4$-bounded degree graphs to $3$-bounded degree graphs.} \label{sec:r2}
Let $H = (V_H, E_H)$ be a $4$-bounded degree graph
and for each 
vertex $v \in V_H$ of degree $4$, denote the neighbours of $v$ by  $w_1, w_2, w_3, w_4$.
Let 
$G = (V, E)$ be the graph obtained from $H$ by the following transformation: for each 
vertex $v \in V_H$ of degree $4$:
substitute $v$ with a three vertex path whose vertices are denoted by $v_a, v_b, v_c$ and add edges 
$(v_a, w_1), (v_a, w_2), (v_c, w_3), (v_c, w_4).$

\medskip

Let ${\cal R}^{(2)}$ be a function that maps a $4$-bounded degree graph $H$ to a $3$-bounded degree graph $G$ according to the procedure defined above.

\begin{figure}[t!]
\vskip 0.1in
\begin{center}
\centerline{\includegraphics[width=0.95\linewidth]{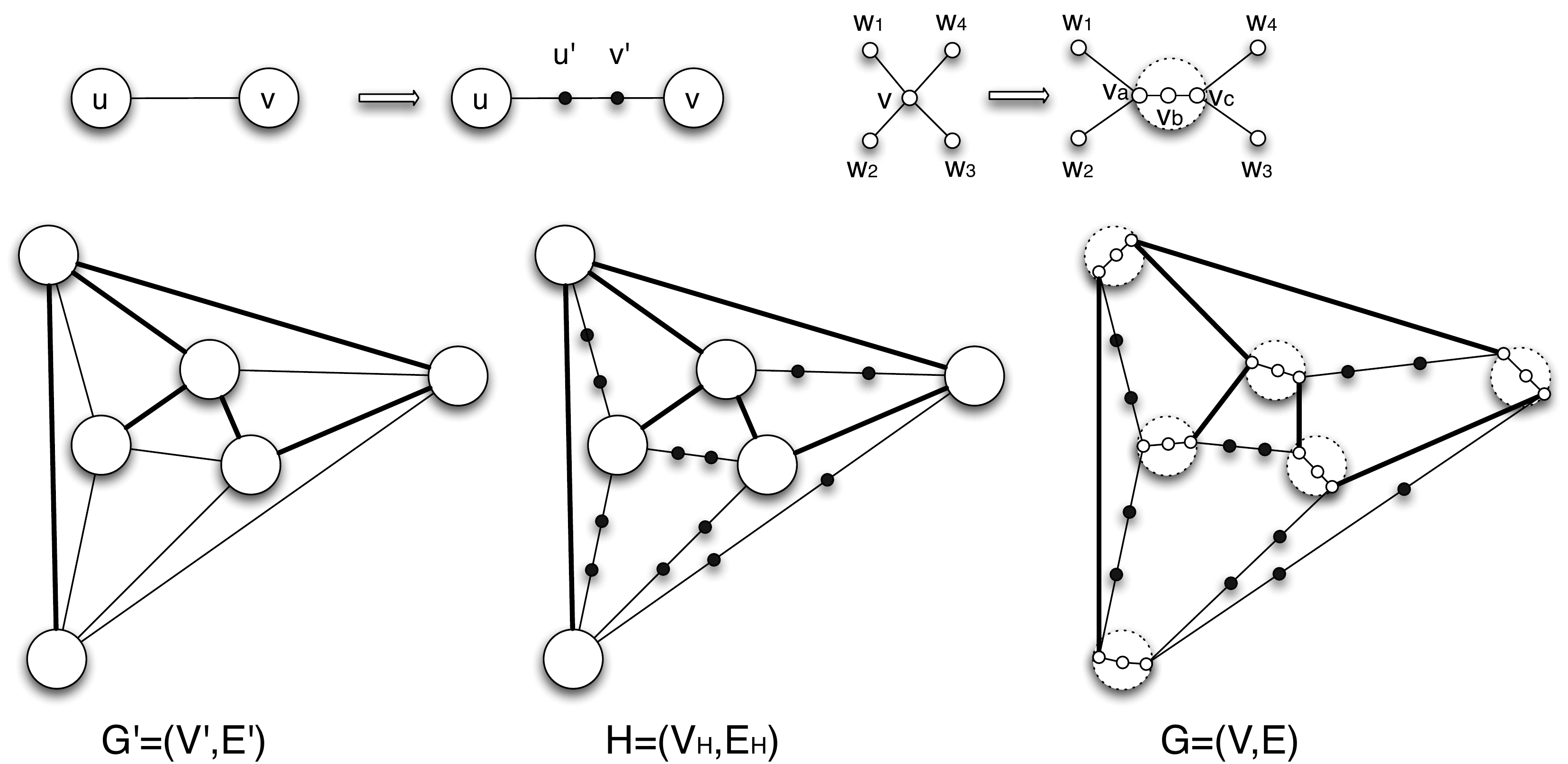}}
\caption{The transformations ${\cal R}^{(1)}$ and ${\cal R}^{(2)}$. On the left the graph $G' = (V, E')$ with the cut set $\hat{E} \subseteq E'$ in bold. 
The center picture shows the corresponding graph $H = {\cal R}^{(1)}(G')$ obtained by the double subdivision of the edges $E' \setminus \hat{E}.$
On the right the graph $G = {\cal R}^{(2)}(H)$ obtained by reducing the degree of the $4$ degree vertices, using the transformation depicted on top.}
\label{fig:reductions}
\end{center}
\vskip -0.2in
\end{figure}

Fig.~\ref{fig:reductions} shows an example of the successive 
application of maps ${\cal R}^{(1)}$ and ${\cal R}^{(2)}$ 
starting from a $4$-regular graph $G'$ and obtaining a $3$-bounded degree and triangle-free graph $G.$

\begin{definition}[3-bounded degree triangle-free hard graphs]
Let $${\cal G}_{{\cal R}^{(1)}, {\cal R}^{(2)}} = \{G = {\cal R}^{(2)}({\cal R}^{(1)}(G')) \mid G' \mbox{ is a $4$-regular graph } \}$$ be the set of 
$3$-bounded degree triangle-free graphs obtained from some 
$4$-regular graph by successively applying functions ${\cal R}^{(1)}$ and ${\cal R}^{(2)}$.
\end{definition}

The following result is a direct consequence of the properties of mappings
${\cal R}^{(1)}, {\cal R}^{(2)}$,  and Theorem \ref{theo:4-reg-graphgap}. 
Its proof is deferred to the appendix.

\begin{theorem} \label{theo:cubicgraphgap}
There are constants $0 < \beta < \beta' < 1$ such that it is $NP$-hard to decide whether 
a graph $G = (V, E) \in {\cal G}_{{\cal R}^{(1)}, {\cal R}^{(2)}}$ 
has a vertex cover of size $\beta |V|$ or all vertex covers of $G$ have size $> \beta' |V|$.
\end{theorem}

\medskip

\bigskip

\noindent \underline{\bf  From graphs in ${\cal G}_{{\cal R}^{(1)}, {\cal R}^{(2)}}$ to 
{\sc PMWIP}$_{Ent}$  instances.} \label{sec:r3}
Let $G = (V, E)$ be a triangle-free and 3-bounded degree graph from ${\cal G}_{{\cal R}^{(1)}, {\cal R}^{(2)}}$, i.e., 
obtained from a 
$4$-regular graph via successively applying functions ${\cal R}^{(1)}, {\cal R}^{(2)}.$ 
Denote by $v_1, \dots, v_n$ the vertices in $V$.
We  construct (in polynomial time) a set of $|E|$  binary vectors $U = \{v^e  \mid e \in E\}$, each of them of dimension $n$, by stipulating that 
if $e = (v_i, v_j)$ then only the $i$-th and $j$-th components of $v^e$ are $1$ and all others are $0$. 

\medskip
We denote by ${\cal R}^{(3)}$ a function that maps each graph in  ${\cal G}_{{\cal R}^{(1)}, {\cal R}^{(2)}}$
to a set of vectors $U$ as described above. 

\medskip

Note that if $E_C$ is a set of edges in a graph $G$ with $C$ being its 
corresponding  set of vectors from  $U = {\cal R}^{(3)}(G)$ 
then we have that $I_{Ent}(E_C) = I_{Ent}(C)$ as 
defined in (\ref{pstarentro}).

\subsection{A gap preserving reduction based on ${\cal R}^{(3)}$}
 \label{sec:gappres-correc}

In this section, focussing on set of vectors $U = {\cal R}^{(3)}(G)$ obtained 
from 3-bounded degree triangle-free graphs 
via the mapping defined in the previous section,  
we  show that there exist
$\eta$ and $\kappa$ for which it is hard to distinguish 
whether  there is a $k$-clustering for $U$  with total impurity
smaller than $\kappa$ or all $k$-clusterings for $U$
have impurity at least $(1+\eta)\kappa$.
In other words, we show that the gap problem
{\sc $\eta$-Gap-PMWIP}$_{Ent}$
 is $NP$-hard and, as a consequence,  by Fact \ref{fact:gap-inapprox}, 
{\sc PMWIP}$_{Ent}$ is $APX$-Hard.

 The proof will consists of two parts: first we show 
how to fix $\kappa$ so that condition 1) of Definition \ref{gap-reduction} 
is satisfied---which we do in 
Lemma \ref{lemma:2-3-star-clustering}
 and Corollary \ref{coro:ck-impurity}; and then we derive $\eta$ so that condition 2) 
of Definition \ref{gap-reduction} is also satisfied---which will be done via a series of propositions 
leading to Lemma \ref{lem:reduction}.

\bigskip

\noindent \underline{\bf  Proving condition 1) of Definition \ref{gap-reduction}.}
The following lemma, which is key for our development, is based on establishing a relationship between minimal vertex covers  and
star decompositions in our hard instances for 3-bounded graphs,
  i.e., graphs from 
${\cal G}_{{\cal R}^{(1)}, {\cal R}^{(2)}}$.

 Given a  graph $G$ and the set of vectors $U = {\cal R}^{(3)}(G)$ 
we  will find it convenient to visualize vectors in $U$ in terms of their corresponding edges 
in $G$. Therefore, for a subset of  $C \subseteq U$ we will 
 say that $C$ is a {\em $p$-star} if the corresponding set of edges 
form a star in $G$, i.e., if $|C| = p$ and there exists a coordinate $j \in [n]$ such that for each vectors $v^e \in C$
the $j$th components of $v^e$ is $1$.

\begin{lemma} \label{lemma:2-3-star-clustering}
Let $G = (V, E)$ be a {\em triangle-free 3-bounded degree} graph in ${\cal G}_{{\cal R}^{(1)}, {\cal R}^{(2)}}$ 
and let $U = {\cal R}^{(3)}(G)$ be the corresponding set of vectors obtained as described in the previous section. 
If $G$ has a minimal vertex cover of size $k$ then there is a $k$-clustering ${\cal C}$ of 
$U$ where  each $C \in {\cal C}$ is either a 2-star or a 3-star.  
\end{lemma}
\begin{proof}
Let $G'=(V', E')$ be the 4-regular graph from which we derived $G$, i.e., $G = {\cal R}^{(2)}({\cal R}^{(1)}(G')).$
Recall that each vertex $v$ of $G'$ is replaced with
3 vertices $v_a$, $v_b$ and $v_c$ in $G$ and 
 there exists a set of $|V'|/2$ edges $\hat{E} \subseteq E'$  such that each edge
$uv \notin \hat{E}$ is replaced with the path $uu'v'v$.  

We first argue that if  $G$ admits 
a minimal vertex cover of size $k$ then there exists a minimal vertex cover 
$S$  for $G$, also with size $k$, that satisfies 
the following properties:

\begin{enumerate}
\item For a vertex $v$ in $G'$ exactly one of the
following conditions holds:  (i) $v_{a},v_{c} \in S$ and $v_{b} \notin S$ or
(ii) $v_{a},v_{c} \notin S$ and $v_{b} \in S$

\item For an edge $e=uv \in E' \setminus  \hat{E}$ exactly one vertex
in the set $\{u',v'\}$
belongs to $S$
\end{enumerate}

To see this, we will show that if $S$ is a minimal vertex cover violating these properties then it is possible
to locally modify $S$ in order to satisfy then both, hence becoming the desired $S$.

Let $S$ be a minimal vertex cover of $G$. Let us first focus on property (1). 
Fix some vertex $v$ from $G'$.  If $v_{b} \notin S$, then we must have $v_{a},v_{c} \in S$, for 
otherwise $S$ does not cover the edges incident in $v_b$. If $v_{b} \in S$ we cannot have that  both $v_{a}$ and $v_{c}$ are in $S$, for otherwise $S$ would not be
minimal since $v_{b}$ can be removed. Hence, if $v_{b} \in S$ at most one
vertex  in $\{v_{a},v_{c} \}$ belongs to $S$. If this is the case, i.e., $S$ contains $v_b$ and exactly one vertex from $\{v_a, v_c\}$,  
we can modify $S$  by replacing $v_{b}$ with
the vertex between $v_a$ and $v_c$ which is not in $S$. The resulting new $S$ is 
a minimal vertex cover with the same cardinality which also satisfies (1).

For property (2), let us fix an edge $uv \in E'\setminus \hat{E},$ and 
let $u, u', v', v$ be the corresponding path in $G$. 
We first note  that at least one of the vertices  
$u', v'$  must be in $S$, for otherwise $S$ does not 
cover edge $u'v'$. In addition, if both belong to $S$ then 
neither $u$ nor $v$ can belong to $S$, for otherwise the cover would not
 be minimal. In this case, however, we can replace 
$v'$ with $v$ obtaining a minimal cover
with the same cardinality and satisfying (2).

In order to compute the desired clustering for $U$ we work
on the graph $G_1$ obtained  from $G$ by undoing the transformation
${\cal R}^{(1)} $
that was employed to remove triangles from the  4-regular graph $G'$.
More precisely, $G_1= {\cal R}^{(2)} (G)$.
Note that $S_1=S \cap V(G_1)$ is  a minimal cover from $G_1$. 
In fact, if  $v_{a} (v_{c})$ is in $S_1$  then 
$v_{b}$ does not belong to $S_1$ (item 1 above)  so that we cannot remove $v_{a} (v_{c})$.
Similarly, if $v_{b} \in S_1$  then $v_{a} \notin S_1$ 
 so that we cannot remove $v_{b}$ from $S_1$.

Now we build a star decomposition for $G_1$ and then we transform it into a star decomposition for $G$.
For every vertex  $v \in S_1$ we will  construct a set $D_1(v)$ consisting
of the edges of the star centred at $v$. 

Let $A=\{\mbox{vertices in } S_1 \mbox{ of degree 3 in } G_1\}$
and  $B=\{\mbox{vertices in } S_1 \mbox{ of degree 2 in } G_1\}$.
Note that $A$ consists of the vertices of type $v_a$ or $v_c$ while
$B$ consists of those of type $v_b$.

Initially, for every $v \in S_1$ we add to $D_1(v)$ 
the edges that connect $v$ to the vertices in $V(G_1) \setminus S_1$.
Since $S_1$ is a minimal cover, after this assignment, we have $|D_1(v)| \ge 1$ for every $v \in S_1$. In addition, we  also have $|D_1(v)|=2$ 
for every node $v \in B$ due to the item 1 above.

We then extend the sets $D_1(v)$ by  applying  the following procedure:

\medskip

\hspace{0.5cm} $E_1 \leftarrow $ edges with both endpoints in $S_1$ (those not yet assigned to a set $D_1(v)$)

\hspace{0.5cm} {\bf While} there exists an edge $e=uv \in E_1$, with  $|D_1(v)|=2$ {\bf do}

\hspace{1.2cm} Remove $e$ from $E_1$ and add it to $D_1(u)$

\medskip

Let $G_1'=(S_1,E_1)$, where $E_1$ is the set of edges 
left unassigned at the end of the above procedure. 
Let  $v \in A $. If $v$ is isolated in $G_1'$ then $|D_1(v)| \ge 2$. Otherwise, if $v$ is non-isolated, then its degree
in  $G_1'$ is 2.
Thus, the non-isolated vertices  in $G_1'$  form
a collection of disjoint cycles. 
Hence, for each $v$ in the cycle we add to $D_1(v)$
exactly one of the two edges incident to it.
This way, we increase the cardinality of each
 $v$   in the cycle by 1 so that  $|D_1(v)| \ge 2$
for every $v \in S_1$, that is, they
are centres of stars of size at least 2.

From the star 
decomposition for $G_1$ we can obtain a family of
sets $\{D(v)\}_{v \in V}$ that will induce a  star decomposition for
 $G$ as follows: initially,  for $v \in A \cup B$, we set $D(v)=D_1(v)$. 
Then, for each  $e=uv \in E' \setminus \hat{E}$ we
proceed as follows: if $uv \in D_1(u)$ 
we create a star centred at $v'$ with edges  $v v'$  and  
$v'u'$ so that $D(v') = \{vv', v'u\}$. In addition, we set $D(u) = (D(u) \setminus e ) \cup  \{u u'\}$.

Then, we define the clustering ${\cal C}$ of vectors in $U$ by creating a cluster for each star in the
decomposition of $G$ and putting in the cluster the vectors corresponding to the edges 
of the star it represents. 
\end{proof}

\begin{corollary} \label{coro:ck-impurity}
Let $G = (V, E)$ be a triangle-free $3$-bounded degree graph from ${\cal G}_{{\cal R}^{(1)}, {\cal R}^{(2)}}$, 
and let $U = {\cal R}^{(3)}(G)$ be the corresponding set of vectors obtained as described in the previous section. 
If $G$ has a minimal vertex cover of size $k$ then there is a $k$-clustering ${\cal C}$ for 
$U$ with impurity $I_{Ent}({\cal C}) = 6k + 3 (|U|-2k) \log 3$.  
\end{corollary}
\begin{proof}
Let  ${\cal C}$ be the $k$-clustering given by Lemma \ref{lemma:2-3-star-clustering} 
and  let $x$ and  $y$ denote the 
number of $2$-stars and $3$-stars in ${\cal C}$, respectively. Then $2x+3y = |U|$ and $x+y = k$, whence, $x = 3k-|U|$ and $y = |U|-2k$. Finally, from 
Fact \ref{fact:star-impurity} we have 
$$I_{Ent}({\cal C}) = 6 x + y \left(6 + 3 \log 3 \right) = 6k + 3 (|U|-2k) \log 3.$$
\end{proof}

The consequence of the last corollary is that when $\mathbb{A}$ denotes the problem of finding the minimum vertex cover of
a graphs in ${\cal G}_{{\cal R}^{(1)}, {\cal R}^{(2)}}$ and 
$\mathbb{B}$ denote the problem {\sc PMWIP}$_{Ent}$, 
the reduction that 
maps a graphs $G = (V,E) \in  {\cal G}_{{\cal R}^{(1)}, {\cal R}^{(2)}}$ to the instance $(U, k)$ of {\sc PMWIP}$_{Ent}$ defined by
$U = {\cal R}^{(3)}$ satisfies 
property 1 in Definition \ref{gap-reduction}, with  $\kappa = 6k + 3 (|E|-2k) \log 3.$

\bigskip

\noindent \underline{\bf  Proving condition 2) of Definition \ref{gap-reduction}.}
We now want to show that when the minimum vertex cover of the graph $G$ has
size at least $k(1+\epsilon)$ then the impurity of every $k$-clustering is at least a constant
times larger than $ 6k + 3 (|U|-2k) \log 3$, which is the impurity of the clustering in Corollary \ref{coro:ck-impurity}, 
which exists when the minimum vertex cover of $G$ has cardinality $\leq k.$ 
This will imply that 
our reduction satisfies also the second property in Definition \ref{gap-reduction}. 

In the following, ${\cal C}$ will denote a clustering of minimum impurity for the instance of 
{\sc PMWIP}$_{Ent}$ obtained 
via the reduction, when, for some constant $\epsilon > 0,$ 
the size of the minimum vertex cover for $G$ is  at least $k(1+\epsilon)$.

We will use the following  notation to describe such a clustering ${\cal C}$ of minimum impurity.

\begin{itemize}
\item $a$: number of clusters in ${\cal C}$ consisting of a 3-star; we refer to these clusters as the $a$-group of clusters;
\item  $b$: number of clusters in ${\cal C}$ consisting of a 2-star; we refer to these clusters as the $b$-group of clusters;
\item  $c$: number of cluster in ${\cal C}$ consisting of a 1-star (single edge); we refer to these clusters as the $c$-group of clusters;
\item  $d$: number of clusters in ${\cal C}$ consisting of  2 edges without common vertex (2-matching); we refer to these clusters 
as the $d$-group of clusters;
\item $e$: number of remaining clusters in ${\cal C}$; we refer to these clusters  as the $e$-group of clusters;
\item $q$: number of edges in the $e$-group of clusters.
\end{itemize}

In the definitions above the letters $a,b,c,d$ and $e$ are used to 
denote both the size and the type of a group of clusters. We believe this
overloaded notation helps the readability.

In Fig.~\ref{fig:smallclusters}, we summarize the impurities of small 
clusters significant to our analyses, according to the above grouping.

\begin{figure}[t!]
\vskip 0.1in
\begin{center}
\centerline{\includegraphics[width=0.95\linewidth]{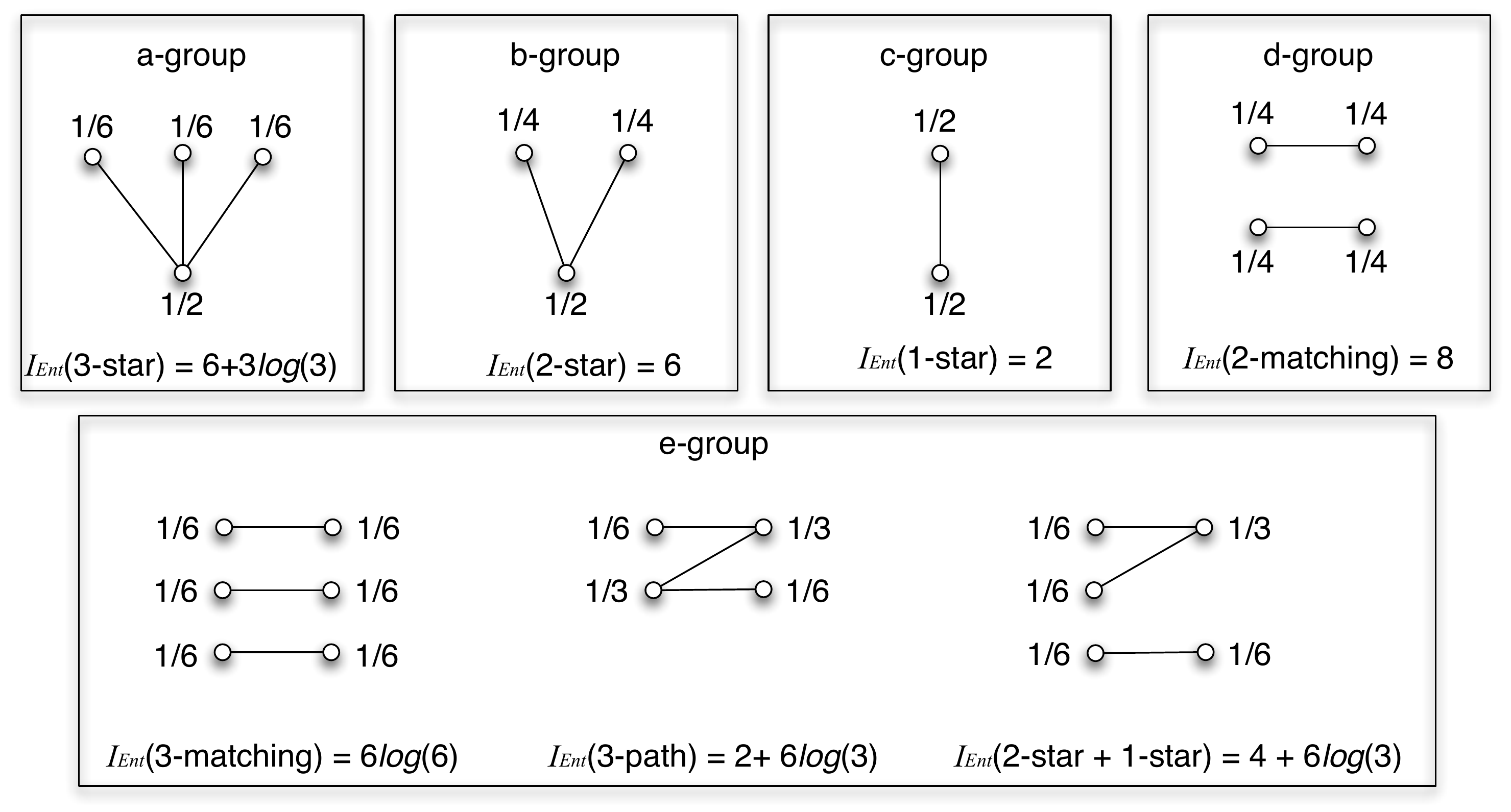}}
\caption{The types of {\em small} clusters $C$ of edges, $|C| \leq 3$ that we can have 
from a triangle-free graph, together with the corresponding impurity $I_{Ent}(C).$
The numbers on the vertices form the distribution ${\bf q}$ such that $I_{Ent}(C) = 2|C| \cdot H({\bf q})$ as in 
(\ref{pstarentro}).}
\label{fig:smallclusters}
\end{center}
\vskip -0.2in
\end{figure}

The following proposition will be useful in our analysis. 

\begin{prop}
Let $x \ge 2$ and let  $n_1$,$n_2$ be   positive integers.
We have that $$n_1(2x+ x \log x)+n_2(2(x+1)+ (x+1) \log (x+1))\ge 
(n_1+n_2)( 2\overline{x} + \overline{x} \log \overline{x}),$$ where  
$\overline{x}= (n_1x+n_2(x+1))/(n_1+n_2)$.
\label{aux:12-11}
\end{prop}
\begin{proof}
It is enough to prove that
 $$n_1 x \log x+n_2 (x+1) \log (x+1)\ge 
(n_1+n_2) \overline{x} \log \overline{x}.$$
This inequality  follows from Jensen inequality since $f(x)=x \log x$ is convex in
the interval $[2,\infty]$ 
\end{proof}

The next two propositions give lower bounds on the sum of the impurities
of the clusters in the $e$-group.

\begin{prop} \label{prop:q-larger-3e}
The total impurity  of the clusters in the $e$-group 
is at least
$ 2q + (q/e) \log (q/e)$.
\end{prop}
\begin{proof}
Let $C_1, \dots, C_e$ be the clusters in the $e$-group. Note that each one of these clusters has cardinality $\geq 3$.
Let $p = \lfloor q/e \rfloor.$ Then, we have 
$p \geq 3.$ Suppose that there exist clusters $C_i, C_j$ such that $|C_i| = x$ and $|C_j| = y$ with 
$y > x+1$. Let $C'_i$ be an $(x+1)$-star and $C'_j$ be a $(y-1)$-star, then we have  
\begin{eqnarray*} 
I_{Ent}(C_i) + I_{Ent}(C_j) &\geq& 2x + x \log x + 2y + y \log y \\
&\geq& 2(x+1) + (x+1) \log (x+1) + 2(y-1) + (y-1) \log (y-1) \\
&\geq& 
I_{Ent}(C'_i) + I_{Ent}(C'_j),
\end{eqnarray*}
where the first inequality follows from Lemma \ref{fact:starLB}, the second inequality holds true for each 
$3 \leq x \leq y-2$ and the last inequality follows from Fact \ref{fact:star-impurity}.

The above inequality says that if we replace $C_i, C_j$ with $C'_i, C'_j$,  the impurity of the resulting set of $e$ clusters 
is not larger than the impurity of the original $e$-group. Moreover, the total number of edges has not changed. 
By repeated application of such a replacement we eventually obtain a group of $e$ clusters 
$\tilde{\cal C} = \{\tilde{C}_1, \dots, \tilde{C}_e\}$ each of cardinality $p$ or $p+1$ and containing in total $e$ edges and such that
$\sum_{i=1} I_{Ent}(C_i) \geq  \sum_{i=1} I_{Ent}(\tilde{C}_i).$ In particular, 
the total impurity of such clusters is not larger than the total impurity of the original $e$ clusters.
Note that these new $e$ clusters need not exist and are only used here for the sake of the analysis. 

Let $n_1$ be the number of clusters in $\tilde{\cal C}$ with $p$ edges and $n_2$ be the number of clusters in $\tilde{\cal C}$ 
with $p+1$ edges and let $\overline{p} = \frac{n_1 p + n_2 (p+1)}{n_1+n_2}.$ 
Then, $q = n_1 p + n_2 (p+1)$ and $e = n_1 + n_2,$ hence $\overline{p} = q/e.$ 

Finally, by applying Proposition \ref{aux:12-11}, we have the desired result:
\begin{eqnarray*} 
\sum_{j=1}^e I_{Ent}(C_i) &\geq& \sum_{j=1}^e I_{Ent}(\tilde{C}_i) \\
&=& n_1 (2p + p \log p) + n_2 (2(p+1) + (p+1) \log (p+1)) \\
&\geq& 
(n_1+n_2) (2\overline{p} + \overline{p} \log \overline{p}) = 2q + \frac{q}{e} \log \frac{q}{e}.
\end{eqnarray*}
\end{proof}

The following fact that can be proved by inspecting
a few cases will be useful.

\begin{fact} 
\label{fact:3-edge}
If a cluster $C$ has 3 edges
and it is neither a $3$-star nor a triangle then
its impurity is at least $2+ 6\log 3$
\end{fact}

\begin{prop} \label{prop:q-smaller-4e}
If $q < 4e$  then the total impurity of the 
clusters in the $e$-group is lower bounded by 
$ 16 (q-3e) + (4e-q) \times (2+ 6\log 3)$
\end{prop}
\begin{proof}
We first consider the case
where every cluster in the $e$-group has either 3 or 4 edges. Let $x$ and $y$ be the number of clusters with
$3$ and $4$ edges, respectively.
Since $x+y=e$ and $3x+4y=q$ we get that
$x=4e-q$ and $y=q-3e$.
It follows from Lemma \ref{fact:starLB} that
the impurity of a cluster with 4 edges is at least
$16$. Moreover, since triangles and $3$-stars are not
allowed in a $e$-group, it follows from Fact \ref{fact:3-edge}
that the impurity of a cluster with 3-edges is at least  
$(2+ 6\log 3)$. Thus, in the case under consideration,
the total impurity is lower bounded by
$16 (q-3e)+(4e-q) \times (2+ 6\log 3)$.

Thus, it suffices to show that the 
$e$-group with $q$ edges and minimum impurity
has only clusters with 3 or 4 edges.
Let us assume for the sake a contradiction that
the $e$-group with minimum impurity has one cluster
$C$ with $r > 4$  edges.
In this case, there exist $r-3$ clusters of cardinality $3$ in the $e$-group, for otherwise the average cardinality would be $\geq 4$, 
violating the assumption $q < 4e$. Let  $D_1, \dots D_{r-3}$ be these clusters.
Moreover, let $C'$ be a $3$-path structure and $D'_1, \dots D'_{r-3}$ be $4$-stars. 
We have 
\begin{equation} \label{samecard}
|C'| + \sum_{j=1}^{r-3} |D'_j| = 3+4(r-3)=r+3(r-3)=|C| + \sum_{j=1}^{r-3} |D_j|.
\end{equation}
Furthermore,  by definition clusters of cardinality $3$ that are in the $e$-group are not $3$-stars.  Thus, by Fact 
\ref{fact:3-edge}
 we have that for each $i=1, \dots r-3,$ it holds that
$I_{Ent}(D_i) \geq 2 + 6 \log 3$ and by Lemma \ref{fact:starLB}, it holds that $I_{Ent}(C) \geq 2r + r\log r.$  
Then   
\begin{eqnarray*}
I_{Ent}(C) + \sum_{j=1}^{r-3} I_{Ent}(D_j) &\geq&  2r + r \log r + (r-3) \left(2 + 6 \log 3\right) \\
&>& \left(2 + 6 \log 3\right) + 16 (r-3)\\
&=& I_{Ent}(C') + \sum_{i=1}^{r-3} I_{Ent}(D'_i), 
\end{eqnarray*}
where the second inequality holds for each $r >4.$

The above inequalities together with (\ref{samecard}) say that if we replace $C , D_1, \dots D_{r-3}$ with $C', D'_1, \dots D'_{r-3}$, 
the impurity of the resulting set of $e$ clusters 
is smaller than the impurity of the original set of $e$ clusters. 
Moreover, the total number of edges does not change. 
Thus, we reach a contradiction,
which establishes the proof.
\end{proof}

\begin{prop} \label{prop:large-cover}
If the minimum vertex cover for $G$ has size at least  $k(1+\epsilon)$ then the following
inequality holds:		
$c+d +q \ge k\epsilon/2 $.
\end{prop}
\begin{proof}
If it does not hold we could construct a vertex cover of size
smaller than $k(1+\epsilon)$ by selecting $a$ vertices to cover the  edges of the  3-stars,
$b$  vertices to cover the edges of the $2$-stars and one vertex per edge
of the other clusters.  The number of edges in these other clusters is  $c+2d+q$.
Hence, we must have $a+b   + c+2d+q \ge k(1+\epsilon) $.
Since $a+b \le k$ we conclude that  $c+d + 2q \ge k \epsilon $,  so that
 $c+d + q \ge c/2+d/2 +q \ge k \epsilon/2$
\end{proof}

Let ${\cal C}^{(k)}$ denote the $k$-clustering only consisting of $2$-stars and $3$-stars 
described in Lemma \ref{lemma:2-3-star-clustering} and Corollary \ref{coro:ck-impurity}, which 
exists when the minimum vertex cover of the graph $G$ has size $\leq k.$

We will now show that, if the minimum size of a vertex cover for $G$ is at least 
$k(1+\epsilon)$ then the impurity of the minimum impurity clustering ${\cal C}$ (for the instance obtained via the reduction) 
is at least a constant factor larger than the impurity  of ${\cal C}^{(k)}.$

\begin{lemma} \label{lem:reduction}
Let $G = (V, E)$ be a  triangle-free and 3-bounded degree graph and
let $U = {\cal R}^{(3)}(G)$ be the corresponding set of vectors obtained as described in section \ref{sec:r3}  above.
If every vertex cover in $G$ has size $\geq k(1+\epsilon),$ then there exists a constant $\eta> 0$ such that
$I_{Ent}({\cal C}) \geq \kappa (1+\eta),$ where
$\kappa = I_{Ent}({\cal C}^{(k)}).$
\end{lemma}
\begin{proof}
Let $E$ be the sum of the impurities of the clusters  in the $e$-group. Then, the impurity of ${\cal C}$ is given by
\begin{equation}
\label{eq:19Aug1}
I_{Ent}({\cal C}) =  (6 +3 \log 3)a+ 2c+ 8d + 6b +  E
\end{equation}
Since $a+b+c+d+e=k$ and $3a+2b+c+2d+q=|U|$ we 
get that $a=c+2e+|U|-q-2k$ and $b=3k-|U|-2c-d-3e+q$.
Replacing $a$ and $b$ in the righthand side of \ref{eq:19Aug1} we get that
$$I_{Ent}({\cal C}) \ge (3\log 3 -4)c + 2d + (6 \log 3 -6)e -(3 \log 3 )q +(3 \log 3) |U| +(6-6 \log 3)k +E $$ 
Thus, 
$$I_{Ent}({\cal C})- I_{Ent}({\cal C}^{(k)})
\ge  (3\log 3 -4)c + 2d + (6 \log 3 -6)e -(3 \log 3) q  + E
$$
Let
$\Delta=I_{Ent}({\cal C}) - I_{Ent}({\cal C}^{(k)})$.
We split the analysis into two cases according to whether  
$\overline{p} = q/e \ge 4$ or $\overline{p} = q/e \le 4$

\noindent
{\bf Case 1} $\overline{p}  \ge 4$.
In this case,

\begin{eqnarray}
\Delta &\ge&  (3 \log 3 -4) c + 2d + [ (\log \overline{p} + 2-3 \log 3)  \overline{p} +6(\log 3 -1)]e  
\label{eq:2-1-end-p1} \\
&=& (3 \log 3 -4) c + 2d +  (  \log \overline{p} +(2-3 \log 3)   +\frac{6(\log 3 -1)}{\overline{p}}  )q    
\label{eq:2-1-end-p2} \\
&\ge& (2.5 - 1.5 \log 3)  ( c+d+q) \label{eq:2-1-end-p3}
\\
&\ge& (1.25 - 0.75 \log 3) k \epsilon
\label{eq:2-1-end-p4}
 \end{eqnarray}
where the  inequality (\ref{eq:2-1-end-p2})-(\ref{eq:2-1-end-p3})  holds
because function  $f(x)=(2-3 \log 3)+\log x +6( \log 3 -1)/x$ is increasing in the interval
$[4,\infty]$ and $f(4)=2.5 - 1.5 \log 3$.
The last inequality follows from Proposition \ref{prop:large-cover}.

\noindent
{\bf Case 2} $\overline{p}  < 4$. 
We have that 

\begin{eqnarray}
\Delta &\ge&  (3 \log 3 -4) c + 2d 
-(3 \log 3) q + (6 \log 3 - 6)e + E \label{z} \\
&\ge & 
(3 \log 3 -4) c + 2d +
(30 \log 3 -46)e + (14-9 \log 3)q
\label{zz} \\
&\ge & 
(3 \log 3 -4) c + 2d +
(2.5 - 1.5 \log 3 )q \label{zzz} \\
&\ge& (2.5-1.5 \log 3)  ( c+d+q) \\
&\ge& (1.25-0.75 \log 3) k \epsilon \label{zzzz} 
 \end{eqnarray}
where  (\ref{zz}) follows from (\ref{z})
due to the lower bound on $E$ given by Proposition \ref{prop:q-smaller-4e}, (\ref{zzz}) follows from (\ref{zz})
because $q < 4e$ and the last inequality follows from 
Proposition \ref{prop:large-cover}. 

From the inequalities
(\ref{eq:2-1-end-p1})-(\ref{eq:2-1-end-p4})
and (\ref{z})-(\ref{zzzz}) we have 
$$
\frac{ I_{Ent}({\cal C})}{I_{Ent}({\cal C}^{(k)})}
\ge \frac{I_{Ent}({\cal C}^{(k)}) +
(1.25-0.75 \log 3) k \epsilon }{I_{Ent}({\cal C}^{(k)})}  
 \ge 1 +  \frac{(1.25-0.75 \log 3) \epsilon}{( 6 + 3 \log 3)},
$$
where the last inequality follows 
because 
$I_{Ent}({\cal C}^{(k)}) = 6k +3(|U|-2k) \log 3$
and $|U| \le 3k$ so that
$I_{Ent}({\cal C}^{(k)}) \ge (6+3 \log 3) k.$
Thus, we can set
\begin{equation} \label{eq:eta-eps}
\eta = \frac{(1.25-0.75 \log 3) \epsilon}{( 6 + 3 \log 3)},
\end{equation}
\end{proof}

\bigskip

\noindent
{\bf Proof of Theorem \ref{theo-main}}. 
Let $0 < \beta <\beta'$ be the constant in Theorem  \ref{theo:cubicgraphgap} and $\epsilon = \frac{\beta'}{\beta} -1.$

By the previous lemma, there exists an $\eta = \eta(\epsilon) >0$ (as in (\ref{eq:eta-eps})) such that 
for every graph $G = (V, E) \in {\cal G}_{{\cal R}^{(1)}, {\cal R}^{(2)}}$ and $k$, 
the set of vectors $U = {\cal R}^{(3)}(G)$ produced according to our reduction is 
such that if $G$ has a vertex cover of size $\leq k$ then $U$ has a $k$-clustering of impurity $\leq \kappa$
and if all vertex covers of $G$ have size $>(1+\epsilon) k$ then all $k$-clustering of $U$ have
impurity $> (1+\eta) \kappa$, 
where $\kappa = 6k + 3 (|E|-2k) \log 3$

This, together with Theorem \ref{theo:cubicgraphgap} implies that 
the {\sc $\eta$-Gap-PMWIP}$_{Ent}$ is NP-hard. Hence, by Fact \ref{fact:gap-inapprox},
if $P \neq NP$ there is no polynomial time $(1+\eta)$-approximation algorithm for 
{\sc PMWIP}$_{Ent}.$
\qed

\medskip

The same arguments in our reduction can also be used to show the inapproximability 
of instances where all vectors have $\ell_1$ norm equal to any constant value, and in particular 1,
i.e., the case where {\sc PMWIP$_{Ent}$} corresponds to {\sc MTC}$_{KL}.$ 
In summary we have the following.

\begin{corollary} \label{coro-apx-MCT}
{\sc MCT}$_{KL}$ is APX-Hard.
\end{corollary}

\section{Conclusions}
\label{sec:conclusions}
In this paper we proved that 
{\sc PMWIP}$_{Ent}$ is APX-Hard even for
the case where all vectors have the same $\ell_1$-norm.
This result implies that $MTC_{KL}$ is APX-Hard resolving
a question that remained open in previous work \cite{cm-fmsitc-08,journals/eccc/AckermannBS11}.
Since there exist logarithmic approximations 
for both {\sc PMWIP}$_{Ent}$ and $MTC_{KL}$ 
(under different parameters) \cite{cm-fmsitc-08,CicLabicml19},
 the main
question that remains open is whether a constant approximation
factor for these problems is possible.

\bibliographystyle{abbrv}
\bibliography{icml}

\appendix{Proof of Theorem \ref{theo:cubicgraphgap}}
\begin{proof}
The proof consists in showing that  we can obtain a gap preserving reduction from {\sc $\epsilon$-Gap-MinVC-4} by 
composition of the function ${\cal R}^{(1)}$ and ${\cal R}^{(2)}$ defined in Section \ref{sec:gappres}. 

Let $G' = (V', E')$ be a $4$-regular graph. Let $n = |V'|$ and $m = |E'|$. Hence $m = 2n.$
Let $H = (V_H, E_H) = {\cal R}^{(1)}(G')$. 
By definition, $H$ has bounded degree $4$ and it is also triangle-free. Let $N = |V_H|$ and $M = |E_H|$. We have 
$N = 3n$ and $M = 4n$.

With reference to the definition of the function ${\cal R}^{(1)}$ in the Section \ref{sec:gappres}, let $\hat{E}$ be the  
set of $n$ edges from $E'$ such that the graph $\hat{G} = (\hat{V}, \hat{E})$ induced by $\hat{E}$ is bipartite and 
$H$ is obtained from $G'$ by a double subdivision of the edges in $E' \setminus \hat{E}.$

\smallskip
\noindent
{\em Observation 1.} Let $A' \subseteq V'$ be a minimal vertex cover of $G'$. 
Then, $A_H = A' \cup \bigcup_{e = (u,v) \in E' \setminus \hat{E}} \{f(e)\}$ with 
$$f(e) = \begin{cases} 
v' & u \in A' \\
u' & u \notin A' 
\end{cases}
$$
is a minimal vertex 
cover of $H$ of size $|A'| + n$.

\smallskip
\noindent
{\em Observation 2.} Let $A_H$ be a minimum size  vertex cover of $H$. 
Then, for each $e = (u,v) \in E' \setminus \hat{E}$ we have that if
both $u'$ and $v'$ are in $A_H$ then neither $u$ nor $v$ are in $A_H$, for otherwise, the minimality of $A_H$ would be violated. 
Let $\tilde{E} = E' \setminus \hat{E}$. Then, the set 
$$A_H' = \left ( A_H \setminus \bigcup_{e = (u,v) \in \tilde{E} \mid u',v' \in A_H} \{u', v'\} \right ) \cup \bigcup_{e = (u,v) \in \tilde{E} \mid u',v' \in A_H} \{u, v'\}$$
 is a vertex cover of $H$ of size equal to $|A_H|$, hence minimum. 
 
Notice that for every edge $e = (u,v) \in \tilde{E}$ exactly one vertex between $u'$ and $v'$ and at least one between $u, v$  is contained in $A_H'.$
Therefore $A' = A_H' \cap V'$ is a vertex cover of $G'$ of size $|A_H'| - n.$ Moreover,  $A'$ is a minimum size vertex cover of $G'$ for otherwise (by Observation 1) 
there would be a vertex cover of $H$ of size smaller than $ |A'| + n = |A_H|.$

\medskip

Putting together the above observations, we have the following.

\medskip
\noindent
{\em Claim 1.} For $0 < \alpha < \alpha'$, if $G'$ has a (minimum) vertex cover of size $\alpha n$ then $H$ has a vertex cover of size $(\alpha+1)n = \frac{\alpha + 1}{3} N$; and if all vertex covers 
of $G'$ have size $> \alpha' n$ then every vertex cover of $H$ has size $> (\alpha' + 1)n =  \frac{\alpha' + 1}{3} N$.

\bigskip
\noindent
Now let $G = (V, E) = {\cal R}^{(2)}(H).$  Since $H$ has $n = N/3$ vertices of degree $4$ and
$G$ is obtained from $H$ by adding, for each vertex of degree $4$ in $H$, two new vertices and two new edges, we have that 
$|V| = |V_H| + 2n = 5n = (5/3)N$ and $|E| = |E_H| + 2n = 6n = 2N = (6/5) |V|.$

\medskip
\noindent
{\em Observation 3.}
Let $C_H$ be a vertex cover of $H$. Let $C_H(4)$ be the set of vertices of $H$ which have degree $4$ and are in $C_H$. 
Let  $U_H(4)$ be the set of vertices of $H$ of degree 4 which are not in $C_H.$ 
Then there is a vertex cover $C$ of $G$ defined by 
$$C = (C_H \setminus C_H(4)) \cup \bigcup_{v \in C_H(4)} \{v_a, v_c\} \cup \bigcup_{v \in U_H(4)} \{v_b\},$$
and $|C| = |C_H| - |C_H(4)| + 2|C_H(4)| + |U_H(4)| = |C_H| + |C_H(4)|+ |U_H(4)| = |C_H| + N/3,$ since the number of vertices in 
$H$ of degree $4$ is equal to $N/3$.

\smallskip
\noindent
{\em Observation 4.} Let $C$ be a minimum size vertex cover of $G$. 
Notice that for each $v \in V_H$ of degree $4$ either $|\{v_a, v_b, v_c\} \cap C| = 2$ or 
$\{v_a, v_b, v_c\} \cap C = \{v_b\}.$ Then, let $V_H(4)$ be the set of vertices in $V_H$ of degree $4$. 
We have that,  setting 
$$C^- =  \bigcup_{\substack{v \in V_H(4) \\ |\{v_a, v_b, v_c\} \cap C| = 2}} ( \{v_a, v_b, v_c\} \cap C) \cup 
\bigcup_{\substack{v \in V_H(4) \\ \{v_a, v_b, v_c\} \cap C = \{v_b\}}} \{v_b\}$$ and 
$$C^+ = \bigcup_{\substack{v \in V_H(4) \\ |\{v_a, v_b, v_c\} \cap C| = 2}} \{v\},$$
the set of vertices 
$$C_H = (C \setminus C^-) \cup C^+$$
 is a vertex cover of $H$ of size  $|C_H| = |C| - |V_H(4)| = |C| - N/3.$

By using the same argument employed before (counter-positive) if all vertex covers of $H$ have size $> c N$, for some constant $0<c<1$,  then every vertex cover of $G$ have
size $> (c+1/3)N.$

 \medskip
 Putting together Claim 1 and Observations 3 and 4 we have the following.
 
 \medskip
 \noindent
 {\em Claim 2.} If $G'$ has a (minimum) vertex cover of size $\alpha n$ then 
       $G$ has a vertex cover of size $\frac{\alpha + 2}{3} N = \frac{\alpha+2}{5} |V|$. 
Moreover, if all vertex covers 
of $G'$ have size $> \alpha' n$ then every vertex cover of 
$G$ has size $> \frac{\alpha' + 2}{3} N = \frac{\alpha'+2}{5} |V|$.

 \medskip
Fix $\beta = \frac{\alpha + 2}{3}$ and $\beta' = \frac{\alpha' + 2}{3}$, with $\alpha, \alpha'$ from Theorem \ref{theo:4-reg-graphgap}.
By Claim 2 and Theorem \ref{theo:4-reg-graphgap}, we have that
there are constants $0 < \beta < \beta' < 1$ such that it is $NP$-hard to decide whether 
a triangle-free $3$-bounded degree graphs $G = (V, E)$ 
in family 
${\cal G}_{{\cal R}^{(1)}, {\cal R}^{(2)}}$
has a vertex cover of size $\beta |V|$ or all vertex covers of $G$ have size $> \beta' |V|$.
The proof is complete. 
\end{proof}

\end{document}